\documentclass[10pt, conference, twocolumn]{IEEEtran}                                                  

\IEEEoverridecommandlockouts                              
\title{Two-hop Communication with Energy Harvesting}

\author{\IEEEauthorblockN{Deniz G\"{u}nd\"{u}z and Bertrand Devillers}
\IEEEauthorblockA{Centre Tecnol\`ogic de Telecomunicacions de Catalunya (CTTC)\\ 08860 - Castelldefels, Barcelona, Spain.}
\thanks{This work is supported in part by EXALTED project (IT-258512) funded by European Union's Seventh Framework Programme (FP7), and by the Spanish Government under project TEC2010-17816 (JUNTOS).}
}

\date{}

\usepackage{amsfonts}
\usepackage{amssymb}
\usepackage{amsmath}
\usepackage{dsfont}
\usepackage{amscd}
\usepackage{graphicx, subfigure}
\usepackage{psfrag}
\usepackage{xcolor}

\usepackage{cite}

\newtheorem{thm}{Theorem}[section]

\newtheorem{lem}[thm]{Lemma}

\begin{document}
\maketitle
\thispagestyle{empty}
\pagestyle{empty}

\begin{abstract}
Communication nodes with the ability to harvest energy from the environment have the potential to operate beyond the timeframe limited by the finite capacity of their batteries; and accordingly, to extend the overall network lifetime. However, the optimization of the communication system in the presence of energy harvesting devices requires a new paradigm in terms of power allocation since the energy becomes available over time. In this paper, we consider the problem of two-hop relaying in the presence of energy harvesting nodes. We identify the optimal offline transmission scheme for energy harvesting source and relay when the relay operates in the full-duplex mode. In the case of a half-duplex relay, we provide the optimal transmission scheme when the source has a single energy packet.
\end{abstract}



\section{Introduction}
\label{sec:intro}

Recent advances in energy harvesting technologies enable self-sustaining wireless nodes that are not limited by the lifetime of a battery. However, identifying the optimal operation of the network with energy harvesting nodes is not a trivial problem and is significantly different from the network lifetime maximization problem in the case of battery limited nodes. The optimal operation of the nodes depends highly on the energy harvesting profile while satisfying the energy-neutral operation \cite{Kansal:ACM:06}, which limits the amount of energy that can be used up to time $t$ by the total harvested energy until that time. This problem has received a lot of interest recently, and several works concerning the optimization of different network utility functions with various assumptions on the available information about the energy profiles of the nodes have appeared in the literature \cite{Kansal:ACM:06, Yang:TC:10, Tutuncuoglu:WC:10, Liu:INFOCOM:10, Castiglione:WIOPT:11, Devillers:JCN:11}.

The problem of minimizing the transmission time for a given amount of data over a point-to-point link with energy harvesting is introduced in \cite{Yang:TC:10}. This problem is the dual to the problem of maximizing the total amount of data that can be transmitted within a deadline constraint, which in turn is closely related to the problem of scheduling packets to minimize the total transmission energy studied in \cite{Uysal:TN:02, Zafer:TN:09}. Optimal transmission in an energy-harvesting system is considered in \cite{Tutuncuoglu:WC:10} when the transmitter has a limited battery size. In \cite{Devillers:JCN:11} a general framework including continuous energy harvesting as well as various battery limitations is provided to optimize the transmission scheme.

This paper studies two-hop transmission in the case of energy harvesting nodes. We focus on a two-hop network composed of an energy-harvesting source, an energy-harvesting relay and a destination. The energy harvesting process at each node is modeled as a packet arrival process, such that each energy packet of a random amount arrives at a random time instant. Our focus will be on offline algorithms, that is, the instants and the amounts of random energy packet arrivals are assumed to be known. Given a deadline constraint of $T$, our goal is to maximize the total amount of data that can be transmitted to the destination within time $T$. The problem, and the difficulty of the solution differs significantly for a full-duplex and half-duplex relay. While the instantaneous transmission power of each terminal over time is to be optimized in both cases; in the case of a half-duplex relay, the transmission schedule between the source and the relay needs to be optimized as well.

%

\begin{figure}
\centering
\psfrag{S}{$S$}\psfrag{R}{$R$}\psfrag{D}{$D$}
\psfrag{Es}{$E^s_i$}\psfrag{Er}{$E^r_i$}
\includegraphics[width=3.0in]{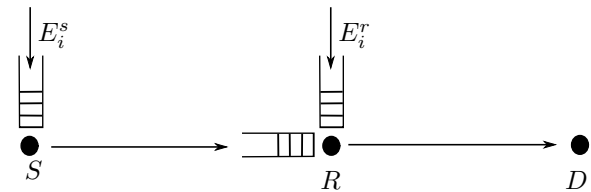}
\caption{Illustration of the two-hop system model with energy-harvesting source and relay terminals.} \label{f:model}
\end{figure}

\section{System Model}\label{s:model}

We consider a two-hop network consisting of a source, a relay and a destination (see Fig. \ref{f:model}). We study the transmission over a time period of $T$. Within time $T$ both the source and the relay harvest energy at known instants of known amounts. Let $0 = t_1 \leq t_2 \leq \cdots \leq t_N$ denote the energy harvesting instants at which the source harvests energy of amounts $E^s_1, \ldots, E^s_N$, respectively, while the relay harvests energy of amounts $E^{r}_1, \ldots, E^{r}_N$, respectively. While we assume identical energy harvesting instants for both nodes; this model is general enough as the time instants at which only the source harvests energy can be modeled by setting $E^{r}_n = 0$, or vice versa. We assume no constraint on the battery sizes. Note here that, different from the previous models we have an energy-harvesting receiver: the relay. While it is interesting to consider the effect of receive energy in this model, our focus in this work is only on the energy used for transmission.


Both the source and the relay can adapt their transmission power and rates instantaneously. We assume additive white Gaussian noise (AWGN) channels from the source to the relay, and from the relay to the destination. Let $h_{s}$ denote the channel coefficient from the source to the relay, while $h_{r}$ denote the channel coefficient from the relay to the destination. The noises at the receivers are assumed to be unit variance and independent of each other. The instantaneous transmission rate from the source to the relay is given by $r^s(P^s(t)) \triangleq \frac12 \log(1+h_{s}P^s(t))$ where $P^s(t)$ is the source power allocated for transmission to the relay at time $t$. Similarly, the rate function from the relay to the destination is given by $r^r(P^r(t)) \triangleq \frac12 \log(1+h_{r}P^r(t))$ where $P^r(t)$ is the power of the relay at time $t$.

The set of feasible power allocation functions at the source and the relay is defined as follows.
\noindent
\begin{small}
\begin{align} \label{d:feasible}
  \mathcal{P} = \bigg \{& (P^s(t), P^r(t)): 0 \leq P^s(t), P^r(t), \forall~t\in[0,T], \nonumber \\
  & \int_{0}^t P^s(t) dt \leq \sum_{k=1}^{n-1} E^s_k \mbox{ and } \nonumber \\
  & \int_{0}^t P^r(t) dt \leq \sum_{k=1}^{n-1} E^{r}_k \mbox{ for all } t\in [t_{n-1}, t_n), \nonumber \\
  & \left. \int_{0}^t r^r(P^r(\tau)) d\tau \leq \int_{0}^t r^s(P^s(\tau)) d\tau~~ \forall t\in [0,T] \right\}.
\end{align}
\end{small}
\noindent While the first two conditions are due to the causality of energy, the last condition is due to the causality in the bits transmitted by the relay; that is, bits can only be transmitted from the relay after they have arrived from the source.

%

The goal is to characterize the maximum number of bits $B$ that can be transmitted to the destination at the end of time $T$ for given energy arrival profiles at the source and the relay. For the solution of this two-hop problem, we often refer to the optimal transmission policy in a single-hop system. As pointed out in \cite{Devillers:JCN:11}, the optimal transmission policy for the single-hop setup, which we call the \textit{Max-Bit algorithm}, can be derived similarly to the dual results in \cite{Uysal:TN:02, Zafer:TN:09}.



\section{Full-duplex Relay}\label{s:full}

We first consider the case with a full-duplex relay. The optimization problem can be written as follows.
\begin{align}
\max~~ \int_0^T r^r(P^r(t)) dt \nonumber \\
\mbox{s.t. } (P^s(t), P^r(t)) \in \mathcal{P}. \nonumber
\end{align}

%


The energy profile at the source limits the bits that can be transmitted to the relay over time. On the other hand, the bits transmitted by the source form a bit arrival profile at the relay. The relay tries to forward as many of its received bits as possible to the destination within time $[0,T]$, satisfying the causality constraints for both the bits and the energy packets. Note that, at the source terminal, the objective function to be optimized is not the total number of bits that can be transmitted by time $T$ as in the single-hop problem in \cite{Yang:TC:10} and \cite{Devillers:JCN:11}. For example, when the number of bits the source transmits to the relay is not the bottleneck, potentially, it is possible for the source to transmit less bits in total, while achieving a different, more advantageous bit arrival profile at the relay, such that the relay can forward more of these bits to the destination. In the following lemma, we show that this is not the case.

\begin{lem}
For given energy arrivals $E^s_{1}, \ldots, E^s_{N}$ at the source, the optimal source transmission scheme is the one that maximizes the total number of bits transmitted to the relay by time $T$. Hence, the optimal transmission scheme is given by the Max-Bit algorithm.
\end{lem}

\begin{proof}
Let $\bar{P}^s(t)$ be the power allocation scheme given by Max-Bit algorithm that maximizes the number of bits transmitted to the relay at time $T$. Let $\bar{E}^s(t)$ and $\bar{B}^s(t)$ be the corresponding cumulative transmitted energy and bit profiles, respectively. That is, we have $\bar{E}^s(t) = \int_0^t \bar{P}^s(\tau) d\tau$ and $\bar{B}^s(t) = \int_0^t r^s(\bar{P}^s(\tau)) d\tau$. We know that $\bar{P}^s(t)$ is a piecewise-linear, non-decreasing function that can increase only at energy arrival instants \cite{Yang:TC:10}. We want to show that any other power allocation profile $\hat{P}^s(t)$ leads to a less advantageous bit arrival profile at the relay.

Note that $\bar{B}^s(t)$ is a piecewise-linear increasing function, changing its slope at time instants, say, $\hat{t}_1=0, \ldots, \hat{t}_{c-1}, \hat{t}_c$. In order to maximize the number of bits the relay forwards to the destination for this bit arrival process under its energy arrival profile, we adapt the algorithm given in \cite{Yang:TC:10} for the scenario when both the data and the energy arrives in packets. In the optimal transmission scheme given by this algorithm, the relay transmits at the highest possible transmission power its energy and bit arrival profiles allow.

Let the transmission scheme $\hat{P}^s(t)$ at the source have a transmitted bit profile $\hat{B}^s(t)$. Consider the bit arrival profile $\check{B}^s(t)$ such that $\check{B}^s(t) = \hat{B}^s(\hat{t}_i)$ for $\hat{t}_{i-1} \leq t < \hat{t}_i$, $i=2, \ldots, c$. We basically let the bits within each interval $[\hat{t}_i, \hat{t}_{i+1})$ arrive at the beginning of the interval. This can only increase the number of bits the relay can forward. Since $\check{B}^s(t)$ is a piecewise linear step function, the algorithm in \cite{Yang:TC:10} will give the optimal relay transmission scheme corresponding to it.

We have $\check{B}^s(t) \leq \bar{B}^s(t)$ at $t = \hat{t}_1, \ldots, \hat{t}_c$, since $\bar{P}^s(t)$ transmits the maximum possible number of bits at time instants of power increase. Now, if we compare the optimal schemes at the relay for $\bar{B}^s(t)$ and $\check{B}^s(t)$, the former will always have a higher degree of freedom since it allows higher rates at the relay; and hence, forward more bits to the destination.
\end{proof}

The maximum number of bits that can be transmitted to the destination can be determined as follows: the source transmits the maximum number of bits that is allowed by its energy arrival profile ignoring the energy profile of the relay, i.e., we use the Max-Bit algorithm at the source. Then the relay, using the bit arrival profile based on source transmission and its own energy arrival profile, forwards as many of these bits as possible to the destination. The optimal transmission scheme at the relay can be derived similarly to the algorithm described in \cite{Yang:TC:10} for a point-to-point system, although the problems are slightly different: here, we have continuous data arrival, rather than in packets, and the goal is not to minimize the transmission time for given data and energy arrivals, but to maximize the number of bits by a given deadline. We do not go into the details of the algorithm due to lack of space.

\section{Half-duplex Relay}\label{s:half}

In the case of a half-duplex relay, we need to optimize the transmission schedule as well as the transmission power of the terminals, that is, in addition to the feasibility conditions for the power allocation functions given in (\ref{d:feasible}), we also have to satisfy $P^s(t)\cdot P^r(t) =0$ for all $t\in [0,T]$.

The problem is significantly harder in the case of a half-duplex relay. For a fixed transmission schedule between the source and the relay, the source does not necessarily use the power allocation scheme that maximizes the number of transmitted bits to the relay at the end of its transmission period. Consider the following example. Both the source and the relay receive energy packets of $E$ units at time $t=0$. Let $T=5t$ for some $t>0$. The source is scheduled to transmit at intervals $[0,t)$ and $[2t, 4t)$, while the relay is scheduled to transmit at intervals $[t,2t)$ and $[4t, 5t)$. We assume $r^s(p) = r^r(p) = r(p)$. Overall, the source transmits over a time period of $3t$. The transmission scheme that maximizes the bits transmitted to the relay allocates a constant power of $E/3t$, and transmits $t\cdot r\left(\frac{E}{3t}\right)$ bits by time $t$, and $3t\cdot r\left(\frac{E}{3t}\right)$ bits by time $4t$. The relay on the other hand would optimally transmit at a constant power of $E/2t$; however, this would require a total of $t\cdot r\left(\frac{E}{t}\right)$ bits available at the beginning of time $t$, which is not the case. Accordingly, the relay transmits at a reduced power of $E/3t$ over the period $[t,2t)$, and transmits at power $2E/3t$ over the period $[4t,5t)$. Overall, the relay can forward a total of $B_1 = t\cdot r\left(\frac{E}{3t}\right) + t\cdot r\left(\frac{2E}{3t}\right)$. On the other hand, if we let the source transmit at a constant power of $E/2t$ over the periods $[0,t)$ and $[2t, 3t)$ and remain silent rest of the time, it transmits $2t\cdot r\left(\frac{E}{2t}\right)$ bits to the relay by time $T$, which is less than the previous scheme. However, the relay can now transmit at power $E/2t$ for a duration of $2t$, and forwards a total of $B_2 = 2t\cdot r\left(\frac{E}{2t}\right)$ bits. We have $B_1 < B_2$ due to the strict concavity of the rate function.

To simplify the problem, we consider the case in which the source terminal receives a single energy packet of $E$ at time $t=0$. In general, we need to optimize the schedule between the source and the relay transmissions. We show for this special case that, in the optimal transmission schedule, first the source transmits a certain amount of bits using all its energy over a time period of $0<t<T$, which needs to be determined, and then the relay forwards all the received bits to the destination in the remaining time period of $T-t$, using the optimal point-to-point transmission schedule.

\begin{lem}\label{l:time_schedule}
In the optimal transmission schedule, the source transmits first and over a connected time interval, and then the relay forwards in the remaining time.
\end{lem}

\begin{proof}
Assume that $P^s(t)$ and $P^r(t)$ are the optimal transmission power assignments for the source and the relay, respectively, such that $P^s(t) \cdot P^r(t)=0$,  $\forall~~t\in [0,T]$. Let
\begin{align}
    t^* \triangleq \int_0^T \mathds{1}(P_s(t)>0)dt,
\end{align}
where $\mathds{1}(x)=1$ if $x$ holds, and $0$ otherwise. Now, define the power allocation function $\hat{P}^s(t)$ as the function that is obtained by shifting the nonzero portions of function $P^s(t)$ to left, that is, $\hat{P}^s(t)$ is nonzero only for $t\in[0,t^*]$. Similarly, we shift the nonzero portions of $P^r(t)$ to right to obtain $\hat{P}^r(t)$, which is nonzero only for $t\in[t^*, T]$. Now, obviously these new functions both satisfy the energy causality constraints since the source has all the energy at time $t=0$, and the relay is only postponing its transmission, and can store the available energy in its battery. On the other hand, the data causality is also satisfied, since the source is transmitting the same amount of data, but earlier than before, which can be stored in the data buffer of the relay. Hence, we do not lose by considering power allocation schemes that divide the time interval $T$ to only two portions, first of which is allocated to the source while the second is allocated to the relay.
\end{proof}

Note that this lemma is not valid when the source also harvests energy, since it might receive energy packets after time $t^*$. Next lemma characterizes the optimal source and relay transmission schemes for given transmission schedule $t^*$.

\begin{lem}
Let $t^*$ be the length of the source transmission period in the optimal transmission schedule. In this period, the source transmits at a constant power of $E/t^*$, while the relay uses the optimal transmission strategy over time $[t^*,T]$ that maximizes the number of bits it transmits until time $T$, while at time $t^*$ it has energy $\sum_{m=1}^{m^*} E^r_m$, where $m^*$ is the maximum $m$ such that $t_m \leq t^*$.
\end{lem}

\begin{proof}
The number of bits that can be forwarded to the destination in this scheme is the minimum of the number of bits the source can transmit to the relay up to time $t^*$ and the number of bits the relay can forward to the destination in the remaining time. The source can transmit at most $t^*\cdot r^s\left(\frac{E}{t^*}\right)$ bits at constant power of $E/t^*$. The maximum number of bits the relay can forward is found by the Max-Bit algorithm.
\end{proof}


We do not have a closed form solution for the optimal transmission time of the source. Even in the case of a single energy packet at the relay, this requires the solution of a non-linear equation. However, since the number of transmitted bits from the source increases from zero to its maximum with increasing $t$, while the number of bits the relay can forward decreases from its maximum to zero, there is always an optimum transmission time $t^* \in (0,T)$, which can be found numerically, that equates these values. While we can use the bisection method to find the crossing point of the optimal bit vs. transmission time curves of the source and the relay, it is possible to simplify the complexity of the algorithm as follows.


Consider the energy arrival profile of the relay. See Fig.~\ref{f:energy_arr_relay} for an example. On this plot, we denote the points corresponding to the instants just before the energy arrivals by $C_1, C_2, \ldots, C_N$, i.e., $C_i$, $i=2, \ldots, N$ is the point $\left(t_i, \sum_{j=1}^{i-1} E^r_j \right)$. While $C_1$ is the origin, the point corresponding to $(T, \sum_{j=1}^N E^r_j)$ is denoted by $C_{N+1}$. We draw the lines connecting $C_{N+1}$ to $C_N, C_{N-1}, \ldots, C_1$. Of these lines, we take the one that does not intersect the energy harvesting curve at points other than $C_i$'s. Let the intersection of it with the $x$-axis be $\tilde{t}_1$. Among the $C_i$'s that this line intersects, we take the leftmost one, say $C_j$, and draw the lines from $C_j$ to all $C_k$ such that $k<j$, and similarly to the previous step we determine $\tilde{t}_2$, and so on so forth, until we reach $C_1$. Using the arguments for the optimal transmission scheme for the energy harvesting relay, if $t^* \in [\tilde{t}_{i+1}, \tilde{t}_i)$, the optimal relay transmission curve first follows the straight line connecting $t^*$ to $C_j$ which is the origin of the straight line that gave us $\tilde{t}_{i+1}$, and continuing on the straight lines drawn by the algorithm until we reach $C_{N+1}$.

This algorithm gives us the maximum number of bits the relay can forward at each point in time. We can calculate this value for points $\tilde{t}_1, \tilde{t}_2, \ldots$ and compare them with the number of bits the source can provide to the relay up to that time instant, which is given by $(1-\tilde{t}_i)r^s_1\left(\frac{E}{1-\tilde{t}_i}\right)$. We stop, say at $\tilde{t}_i$, when the value for the number of bits the relay can forward is more than the one the source can transmit. The optimal transmission time for the source $t^*$ lies between $\tilde{t}_i$ and $\tilde{t}_{i-1}$. We use the bisection method in this range to find an accurate approximation of $t^*$.

For the sample energy arrival curve given in Fig. \ref{f:energy_arr_relay}, we set $E_1=5$, $E_2=5$, $E_3=6$, $t_2=7$, $t_3=10$ and $T=11$. We also fix the channel coefficients as $h_s = h_r = 1$. In Fig. \ref{f:bit_energy}, we plot the maximum number of bits that can be transmitted to the destination vs. the energy of the source. As illustrated in the figure with different colors, as the source energy increases, the optimal time allocation $t^*$ decreases, and it falls into different time intervals in the relay energy arrival curve. It can also be seen in the figure that the number of forwarded bits is bounded above by the maximum number of bits the relay can forward even if it had an unlimited bit supply.

\begin{figure}
\centering
\psfrag{E1}{$E_1$}\psfrag{E2}{$E_2$}\psfrag{E3}{$E_3$}
\psfrag{c1}{$C_1$}\psfrag{c2}{$C_2$}\psfrag{c3}{$C_3$}\psfrag{c4}{$C_4$}
\psfrag{t1}{$t_1$}\psfrag{t2}{$t_2$}\psfrag{t3}{$t_3$}
\psfrag{tt1}{$\tilde{t}_1$}\psfrag{tt2}{$\tilde{t}_2$}\psfrag{tt3}{$\tilde{t}_3$}
\psfrag{T}{$T$}
\psfrag{energy}{$\mathrm{Energy}$}\psfrag{time}{$\mathrm{Time}$}
\includegraphics[width=3.3in]{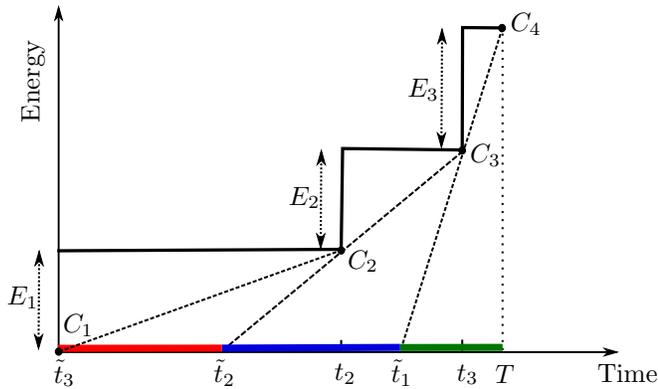}
\caption{The cumulative energy arrival profile at the relay.} \label{f:energy_arr_relay}
\end{figure}

\section{Conclusions and Future Work}\label{s:discuss}

We have introduced the two-hop network model in which the source and the relay terminals are capable of harvesting energy from the environment. We have focused on the offline algorithms which optimize the transmission schemes with the knowledge of the energy harvesting profiles of all the nodes. We have provided the optimal transmission scheme in the case of a full-duplex relay. On the other hand, in the case of a half-duplex relay, characterizing the optimal transmission schedule in the most general case is quite complicated. Instead, we have focused on the special case of a single energy packet at the source terminal. For this simplified scenario, we have identified the characteristics of the optimal transmission schemes at the source and the relay terminals, and we have provided a low complexity numerical algorithm that solves for the optimal transmission time schedule.



An interesting extension of this work that deserves further analysis is the model with battery or data buffer limitations at the relay terminal. Note that, when the battery or the data buffer of the relay terminal is limited in the half-duplex case, Lemma \ref{l:time_schedule} does not hold since either the energy or the data accumulated over the time period when the source is transmitting will overflow. Instead, we should optimize the transmission schedule taking into account the battery or the data overflow at the relay terminal.

\begin{figure}
\centering
\psfrag{B}{$B$}\psfrag{Es}{$E$}
\includegraphics[width=3.3in]{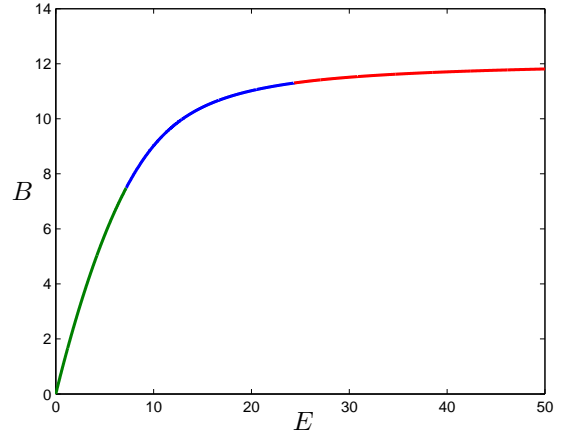}
\caption{The maximum number of transmitted bits, $B$, versus the source energy, $E^s$ for the energy arrival profile given in Fig. \ref{f:energy_arr_relay}.} \label{f:bit_energy}
\end{figure}

\bibliographystyle{IEEEtran}
\bibliography{ref2}

\begin{thebibliography}{1}
\providecommand{\url}[1]{#1}
\csname url@samestyle\endcsname
\providecommand{\newblock}{\relax}
\providecommand{\bibinfo}[2]{#2}
\providecommand{\BIBentrySTDinterwordspacing}{\spaceskip=0pt\relax}
\providecommand{\BIBentryALTinterwordstretchfactor}{4}
\providecommand{\BIBentryALTinterwordspacing}{\spaceskip=\fontdimen2\font plus
\BIBentryALTinterwordstretchfactor\fontdimen3\font minus
  \fontdimen4\font\relax}
\providecommand{\BIBforeignlanguage}[2]{{%
\expandafter\ifx\csname l@#1\endcsname\relax
\typeout{** WARNING: IEEEtran.bst: No hyphenation pattern has been}%
\typeout{** loaded for the language `#1'. Using the pattern for}%
\typeout{** the default language instead.}%
\else
\language=\csname l@#1\endcsname
\fi
#2}}
\providecommand{\BIBdecl}{\relax}
\BIBdecl

\bibitem{Kansal:ACM:06}
A.~Kansal, J.~Hsu, S.~Zahedi, and M.~B. Srivastava, ``Power management in
  energy harvesting sensor networks,'' \emph{ACM Trans. Embedded Computing
  Sys.}, vol.~6, no.~4, Sep. 2006.

\bibitem{Yang:TC:10}
J.~Yang and S.~Ulukus, ``Optimal packet scheduling in an energy harvesting
  communication system,'' \emph{IEEE Trans. on Communications},
  \textit{submitted}, June 2010.

\bibitem{Tutuncuoglu:WC:10}
K.~Tutuncuoglu and A.~Yener, ``Optimum transmission policies for battery
  limited energy harvesting nodes,'' \emph{submitted to IEEE Trans. Wireless
  Comm.}, Sep. 2010.

\bibitem{Liu:INFOCOM:10}
R.~Liu, P.~Sinha, and C.~E. Koksal, ``Joint energy management and resource
  allocation in rechargeable sensor networks,'' in \emph{Proc. IEEE INFOCOM},
  San Diego, CA, Mar. 2010.

\bibitem{Castiglione:WIOPT:11}
P.~Castiglione, O.~Simeone, E.~Erkip, and T.~Zemen, ``Energy-neutral
  source-channel coding in energy-harvesting wireless sensors,'' in \emph{Proc.
  WiOpt}, Princeton, NJ, May 2011.

\bibitem{Devillers:JCN:11}
B.~Devillers and D.~G\"{u}nd\"{u}z, ``A general framework for the optimization
  of energy harvesting communication systems with battery imperfections,''
  \emph{submitted to Journal of Commun. and Networks}.

\bibitem{Uysal:TN:02}
E.~Uysal-Biyikoglu, B.~Prabhakar, and A.~E. Gamal, ``Energy-efficient packet
  transmission over a wireless link,'' \emph{IEEE/ACM Trans. Networking},
  vol.~10, no.~4, pp. 487--499, August 2002.

\bibitem{Zafer:TN:09}
M.~A. Zafer and E.~Modiano, ``A calculus approach to energy-efficient data
  transmission with quality-of-service constraints,'' \emph{IEEE/ACM Trans.
  Networking}, vol.~17, no.~3, pp. 898--911, June 2009.

\end{thebibliography}

\end{document}